%% file: root.tex
\documentclass{ifacconf}

\usepackage{graphicx}      
\usepackage{natbib}        

\usepackage{subfig}
\usepackage[T1]{fontenc}

\usepackage{xcolor}

\usepackage{amsmath}
\usepackage{amssymb}
\usepackage{dirtytalk}
\usepackage{url}

\usepackage{amssymb}
\makeatletter
\newcommand*{\transpose}{%
	{\mathpalette\@transpose{}}%
}
\newcommand*{\@transpose}[2]{%
	\small \raisebox{1.37\depth}{$\m@th#1\intercal$}%
}
\makeatother

\newcommand{\R}{\mathbb{R}}
\newcommand{\Complex}{\mathbb{C}}
\newcommand{\abs}[1]{\left| #1 \right|}
\newcommand{\norm}[1]{\left| #1 \right|}
\newcommand{\norminfty}[1]{\norm{#1}_{\infty}}

\newcommand{\normltwo}[1]{\norm{#1}}

\newcommand{\normLtwo}[1]{\norm{#1}_{L^{2}}}

\newcommand{\normHtwo}[1]{\norm{#1}_{H^{2}}}

\newcommand{\X}{\mathcal{X}}

\newcommand{\T}{\mathcal{T}}
\newcommand{\Tinv}{\mathcal{T}^*}
\newcommand{\tcut}{t_c}
\newcommand{\param}{\theta}
\newcommand{\invparam}{\theta}
\newcommand{\Tparam}{\T_\param}
\newcommand{\Tinvparam}{\Tinv_\invparam}
\newcommand{\noise}{\epsilon}
\newcommand{\omegac}{\omega_c}
\newcommand{\Dparam}{\omegac}

\newtheorem{assumption}{Assumption}

\newtheorem{theorem}{Theorem}
\newtheorem{remark}{Remark}
\newtheorem{proposition}{Proposition}

\def\O{\mathcal{O}}

\def\X{\mathcal{X}}

\begin{document}
\begin{frontmatter}
	
\title{Towards Gain Tuning for Numerical KKL Observers\thanksref{footnoteinfo}} 

\thanks[footnoteinfo]{This work was supported by Ansys Inc.}

\author[CAS,DSME,Ansys]{Mona Buisson-Fenet} 
\author[DSME]{Lukas Bahr} 
\author[Ansys]{Valery Morgenthaler} 
\author[CAS]{Florent Di Meglio}

\address[CAS]{Centre Automatique et Systèmes (CAS), Mines Paris -- PSL, Paris, France (e-mail: mona.buisson@minesparis.psl.eu)}
\address[DSME]{Institute for Data Science in Mechanical Engineering, RWTH Aachen University, Aachen, Germany}
\address[Ansys]{Ansys Research Team, Ansys France, Villeurbanne, France}

\begin{abstract}                
This paper presents a first step towards tuning observers for general nonlinear systems.
Relying on recent results around Kazantzis-Kravaris/Luenberger (KKL) observers, we propose an empirical criterion to guide the calibration of the observer, by trading off transient performance and sensitivity to measurement noise.
We parametrize the gain matrix and evaluate this criterion over a family of observers for different parameter values.
We then use neural networks to learn the mapping between the observer and the nonlinear system, and present a novel method to sample the state-space efficiently for nonlinear regression. 
We illustrate the merits of this approach in numerical simulations. 
\end{abstract}

\begin{keyword}
	Nonlinear observers and filter design; Continuous time system estimation; Machine learning; Estimation and filtering; Observer design.
\end{keyword}
	
\end{frontmatter}

\input{sections/0_introduction}

\input{sections/1_kkl_observers}
\input{sections/2_criterion}
\input{sections/3_numerical_methods}

\input{sections/4_results}
\input{sections/6_conclusion}

\begin{ack}
The authors would like to thank Sebastian Giedyk at the Institute for Data Science in Mechanical Engineering (RWTH Aachen University) for his help with the experimental data.
Thanks also to Pauline Bernard, Philippe Martin and Laurent Praly for all the fruitful discussions around this paper.
\end{ack}

\bibliography{library}             

\end{document}

%% file: sections/0_introduction.tex
\section{Introduction}

In this paper, we propose a numerical method to calibrate Kazantis-Kravaris-Luenberger (KKL) observers. 
The original design of Luenberger observers for linear systems can be found in~\citet{Luenberger1966}. It consists in finding a linear mapping between the system dynamics and a linear filter of the measurement. Under appropriate observability assumptions and filter design, the Sylvester equation satisfied by the mapping has a unique injective solution. Its left-inverse, along with the filter, can be used to compute a convergent state estimate. 

This design encompasses important degrees of freedom: the matrices defining the filter or, equivalently, the poles and zeros of the filter transfer function. To study their effect on state estimation performance, one must consider the effect of the mapping, which modifies the response, among others to measurement noise. For autonomous linear systems, the problem of tuning these degrees of freedom is essentially solved by the stationary Kalman Filter~\citep{Kalman1961}. Rather than directly assigning closed-loop eigenvalues, one can weigh the relative confidence in the measurement and the dynamic model and find the observer gains that are optimal for the metric defined by these weights.

The extension of these approaches to nonlinear systems is nontrivial.
Indeed, there are few generic nonlinear observer designs; a review of these can be found in~\citet{Pauline_observer_design_nonlin_systs,Pauline_survey_observers}. Among the most commonly used are the High-Gain Observer (HGO)~\citep{Bornard1991, Khalil2014} and the Extended Kalman Filter (EKF)~\citep{Gelb1974}. The EKF consists in linearizing the observer dynamics around the current estimate to compute the optimal gain depending on chosen weights, akin to the linear case. There are, however, only local convergence guarantees~\citep{Krener2003}. Conversely, the HGO relies on a change of variables to bring the system into canonical form, and high gains to \say{dominate} the Lipschitz constant of the nonlinearity. The stability guarantees come at the price of possibly poor transient performance (the so-called \say{peaking} phenomenon~\citep{maggiore2003separation}) and high sensitivity to noise. While recent contributions aim at reducing these detrimental features thanks, e.g., to dynamic extensions~\citep{Astolfi2018}, the question of gain tuning and performance criteria remains open. In particular, in~\citet{Astolfi2018}, the sensitivity to noise is examined a posteriori through numerous simulations.

In this paper, we develop a tuning methodology for Kazantzis-Kravaris/Luenberger (KKL) observers that does not rely on extensive tests, inspired by the Kalman filter or $H_\infty$ control. The KKL design~\citep{Kazantzis1998, Praly_existence_KKL_observer} extends the results of~\citet{Luenberger1966} to nonlinear systems. It maps the system dynamics to a stable linear filter of the measured output, called the observer dynamics. The existence and injectivity of this mapping are guaranteed by mild observability conditions, which makes this design relatively generic. The contraction properties of the observer dynamics ensure convergence of the state estimates. The main challenge consists in computing said mapping and its left inverse, along with tuning the free parameters of the observer.

In~\citet{Ramos2020}, a method is proposed to approximate the mapping by performing nonlinear regression on datasets generated from trajectories of the system and the observer. Given fixed observer parameters, a neural network approximates the considered mapping, which is then used to compute state estimates from observer values. 

In this paper, we build on the approach of~\citet{Ramos2020} and propose a first step towards calibration of the observer. 
Our main contribution is a procedure to select the gain matrix using a tuning criterion that, in some sense, trades off the transient performance against the sensitivity to measurement noise.
We start by setting this matrix based on a pre-defined filter, parametrized by its cut-off frequency.
We then approximate the KKL mapping for different values of this parameter using neural networks, either independently or by learning the mapping as a function of the parameter.
This approximation is enabled by appropriately sampling the state-space, improved upon~\citet{Ramos2020}.
Computing the proposed criterion for all values of the parametrized gain matrix leads to an optimal calibration for the observer, in the sense of the proposed empirical criterion.
Numerical simulations illustrate the approach. 

The paper is organized as follows. In Sec.~\ref{sec:KKL_observers}, we recall the main idea behind KKL observer design. 
In Sec.~\ref{sec:gain_tuning_criterion}, we propose an empirical gain tuning criterion, then detail our numerical approach for state-space sampling, observer parametrization and nonlinear regression in Sec.~\ref{sec:numerical_methods}. Finally, we illustrate the merits of the approach through numerical simulations in Sec.~\ref{sec:results}, before concluding in Sec.~\ref{sec:conclu}.

%% file: sections/1_kkl_observers.tex
\section{KKL observers}
\label{sec:KKL_observers}

Consider the autonomous nonlinear dynamical system
\begin{align}
	\begin{aligned}
		\dot x &= f(x)   \\
		y &= h(x) 
	\end{aligned}\label{eq:main_dynamics}
\end{align}
where $x \in \R^{d_x}$ is the state, $y \in \R^{d_y}$ is the measured output, $f$ is a continuously differentiable function $(C^1)$ and $h$ is a continuous function. The goal of observer design is to compute an estimate of the state~$x(t)$ from the knowledge of the past values of the output $y(s)$, $0 \leq s \leq t$. To ensure the feasibility of this task, KKL observers rely on the following two assumptions.
\begin{assumption}
	\label{ass_bounded}
	There exists a compact set $\X$ such that for any solution of interest $x$ to \eqref{eq:main_dynamics},  $x(t)\in \X$  for all $t\geq 0$.
\end{assumption}
\begin{assumption}
	\label{ass_detect}
	There exists an open bounded set $\O$ containing $\X$ such that \eqref{eq:main_dynamics} is \textit{backward $\O$-distinguishable} on $\X$, namely for any trajectories $x_a$ and $x_b$ of \eqref{eq:main_dynamics} such that $(x_a(0),x_b(0))\in \X\times \X$ and $x_a(0)\neq x_b(0)$, there exists $t < 0 $ such that
	$$
	h(x_a(t))\neq h(x_b(t)) 
	$$
	and $(x_a(\tau),x_b(\tau))\in \O\times \O$ for all $\tau\in[t,0]$.
	In other words, their respective outputs become different in backward finite time before leaving $\O$.
\end{assumption}
We now recall the following Theorem from~\citet{Praly_existence_KKL_observer} showing the existence of a KKL observer. 
\begin{theorem}[\citet{Praly_existence_KKL_observer}]
	\label{theo_auto}
	Suppose Assumptions \ref{ass_bounded} and \ref{ass_detect} hold. Define $d_z = d_y(d_x+1)$. Then,  there exists $\ell >0$ and a set $S$ of zero measure in $\mathbb{C}^{d_z}$ such that for any diagonalizable matrix $D\in \Complex^{d_z\times d_z}$ with eigenvalues $(\lambda_1,\ldots,\lambda_{d_z})$ in $\mathbb{C}^{d_z}\setminus S$ with $\Re \lambda_i <-\ell$, and any~$F\in \Complex^{d_z\times d_x}$ such that $(D,F)$ is controllable, there exists a continuous injective mapping~$\T : \R^{d_x} \rightarrow \Complex^{d_z}$ that satisfies the following equation on $\X$
	\begin{align}
		\frac{\partial  \T}{\partial x}(x)f(x) = D\mathcal  \T(x)+Fh(x),\label{eq:pdeT}
	\end{align}
	and its continuous pseudo-inverse~$ \Tinv: \Complex^{d_z} \rightarrow \R^{d_x}$ such that the trajectories of~\eqref{eq:main_dynamics} remaining in $\X$ and any trajectory of
	\begin{align}
		\dot z &= D z + F y
		\label{eq:dynz}
	\end{align}
	satisfy
	\begin{align}
		\left|z(t) -  \T(x(t)) \right| \leq M \left|z(0) -  \T(x(0)) \right|e^{-\lambda_{\min} t} \label{eq:expBound}
	\end{align}
	for some~$ M > 0$ and with~
	\begin{align}
		\lambda_{\min} = \min \left\{|\Re \lambda_1|,\ldots, |\Re\lambda_{d_z}|\right\} \label{eq:lambdaMin} .
	\end{align}
	Due to the uniform continuity of $\Tinv$, this yields:
	\begin{align}
		\lim_{t \to +\infty}\left|  \Tinv(z(t)) - x(t) \right| = 0.
	\end{align}
\end{theorem}

Note that according to this result, $z \in \Complex^{d_y(d_x+1)}$. Therefore, in order to represent this filter with real numbers only, we need $d_z = 2 d_y(d_x+1)$. However, in practice we assume that the $d_y(d_x+1)$ complex eigenvalues needed for $D$ are complex conjugates, such that we only need dimension $d_z = d_y(d_x+1)$ to represent the real filter $z \in \R^{d_z}$.

Thus, implementing a KKL observer involves following the steps:
\begin{enumerate}
	\item Choose matrices $D$ and $F$
	\item Compute the corresponding transformation $\Tinv$
	\item Simulate \eqref{eq:dynz} from an arbitrary $z(0)$ and compute the estimate $\hat x(t) = \Tinv(z(t))$.
\end{enumerate}

In \citet{Ramos2020}, a method to complete step 2 by performing nonlinear regression on trajectories of~\eqref{eq:main_dynamics} and~\eqref{eq:dynz} is proposed. In the next section, we propose an approach to assist the user in completing step 1 by defining a performance criterion to optimize. 

%% file: sections/2_criterion.tex
\section{A gain tuning criterion}
\label{sec:gain_tuning_criterion}

Consider the dynamical system~\eqref{eq:main_dynamics} and associated observer dynamics~\eqref{eq:dynz}. 
Denote $x$,~$z$ their solutions starting respectively at $x(0)$ and~$\T(x(0))$.
Assume now that the measurement $y$ is corrupted by an unknown noise vector $\noise \in \R^{d_y}$, so that $y(t) = h(x(t)) + \noise(t)$.
Denote $\hat{z}$ the corresponding solution of~\eqref{eq:dynz} starting at an arbitrary initial condition $z_0$, and $\tilde{z} = \hat{z} - z$ the estimation error due to both the initial error and the measurement noise. 
In general, we aim to choose $D$ such that the overall error on the estimated state $\hat{x}$ is minimized, where $\hat{x} = \Tinv(\hat{z})$, similarly to~\citet{Henwood_PhD_estimation_online_parametres_machines_electriques_suivi_temperature_composants}. The following result provides a criterion for tuning $D$, which we then apply to the approximated transformation.

\begin{proposition}
\label{prop:boundHatx}
Suppose Assumptions~\ref{ass_bounded} and~\ref{ass_detect} are verified, such that Theorem~\ref{theo_auto} holds. Further, assume that $\Tinv$ is Lipschitz continuous of constant $L$.
Then, we have
\begin{align}
    \normLtwo{ \hat{x} - x } \leq L \Big( \norminfty{G_\noise} \normLtwo{\noise} + \normHtwo{G_z} \normltwo{\tilde{z}(0)} \Big) 
    \label{eq:boundHatx}
\end{align}
where $\normltwo{\cdot}$ is the Euclidean norm, $\normLtwo{\cdot}$ is the $L^2$ norm, and the $H^2$ respectively $H_\infty$ norms are defined as 
\begin{align}
	\normHtwo{G}^2 = \frac{1}{2\pi} \int_{-\infty}^{\infty} \normltwo{G(j\omega)}^2d\omega , \quad \norminfty{G} = \sup_\omega \abs{G(j \omega)}
\end{align}
with $G_\noise(s) = (s I_{d_z} - D)^{-1} F$ the transfer function from $\noise$ to $\tilde{z}$, and $G_z(s) = (s I_{d_z} - D)^{-1}$ from $\tilde{z}(0)$ to $\tilde{z}$.
\end{proposition}

\begin{pf}
By Lipschitz continuity of $\Tinv$, we have
\begin{align}
    \normLtwo{\hat{x} - x} ^2 & =\int_0^\infty \normltwo{ \Tinv(\hat{z}(t)) - \Tinv(z(t)) }^2 dt
    \notag
    \\
	& \leq L ^2 \normLtwo{\tilde{z}}^2 .
	\label{eq:bound_hatx_lipsch}
\end{align}
The Laplace transform applied to the dynamics of $\tilde{z}$ yields
\begin{align}
    \underline{\tilde{z}}(s) & = (s I_{d_z} - D)^{-1} F \underline{\noise}(s) + (s I_{d_z} - D)^{-1} \tilde{z}(0)
    \notag
    \\
    & = G_\noise(s) \underline{\noise}(s) + G_z(s) \tilde{z}(0) ,
\end{align}
where we denote the Laplace transform of a signal $f(t)$ by $\underline{f}(s)$. Applying standard results on signal norms for linear systems \citep{Toivonen_signal_system_norms} yields
\begin{align}
    \normLtwo{\tilde{z}} & = \normLtwo{\underline{\tilde{z}}} \leq \norminfty{G_\noise} \normLtwo{\noise} + \normHtwo{G_z} \normltwo{\tilde{z}(0)} .
    \label{eq:bound_ztilde}
\end{align}
Replacing \eqref{eq:bound_ztilde} in \eqref{eq:bound_hatx_lipsch} concludes the proof.
\end{pf}

Proposition~\ref{prop:boundHatx} exhibits a standard trade-off in linear system theory, between sensitivity to noise through the term in $\normLtwo{\noise}$ and convergence speed through the term in $\normltwo{\tilde{z}(0)}$. 
In this paper, we propose a heuristic that guides the choice of $D$ such that the error on the estimate $\hat{x}$ is minimized.

\begin{remark}
    Proposition \ref{prop:boundHatx} relies on the assumption that $\Tinv$ is Lipschitz continuous. This is not true in general; however, we approximate $\Tinv$ with the neural network model $\Tinvparam$, which is Lipschitz if its activation function is and if its weights are bounded \citep{Scaman_Lipschitz_deepNN_AutoLip}. Its Lipschitz constant can be approximated empirically, for example by computing its maximum over a regular grid of $n$ samples $z_j$. However, the maximum value is subject to outliers and tends to vary strongly between models.
\end{remark}

In the light of this remark, we propose to monitor the following empirical criterion
\begin{align}
\label{eq:gain_tuning_criterion}
	\alpha(D) & := \normltwo{J} \big( \norminfty{G_\noise} + \normHtwo{G_z} \big) 
	\notag
	\\ 
	J & := \bigg( \normltwo{\frac{\partial \Tinv}{\partial z}  (z_j)} \bigg)_{j \in \{1,\cdots, n\}} 
\end{align}
where we consider the $l_2$-norm of $J$ rather than its infinity norm.
This is an approximate bound for $\normLtwo{\hat{x} - x}$. 
This heuristic trades off the transient through $\normHtwo{G_z}$, and the performance and noise sensitivity through $\norminfty{G_\noise}$ and $\normltwo{J}$.
In our experiments, we consider a family of matrices $D$ indexed by a scalar parameter~$\Dparam$. We then compute $\alpha$ for different $D(\Dparam)$ and pick the value of $\Dparam$ that minimizes it. 

\begin{remark}
	The bound~\eqref{eq:bound_hatx_lipsch} is conservative, and the choice of the $L^2$ norm is somewhat arbitrary. In practice, one could consider a variety of criteria by weighting different norms of~$\frac{\partial \mathcal T^*}{\partial z}$, $G_\noise$ and $G_z$. 
	For example, in the linear case where $\T$, $\Tinv$ are matrices, we have
	\begin{align}
	    \underline{\hat{x}}(s) -\underline{x}(s) & = (sI_{d_z} - \Tinv D \T)^{-1} \Tinv F \underline{\noise}(s) 
         \notag
         \\ & + (sI_{d_z} - \Tinv D \T)^{-1}(\hat{x}(0) - x(0)).
         \label{eq:linear_criterion}
	\end{align}
	Another criterion could be the $H_\infty$ norm of an analogy of this transfer function~\eqref{eq:linear_criterion} for the nonlinear case using the empirical gradients.
	Note also that there are more advanced methods to estimate the Lipschitz constant of $\Tinvparam$~\citep{Scaman_Lipschitz_deepNN_AutoLip}; we focus on the simpler criterion~\eqref{eq:gain_tuning_criterion}, which is enough to exhibit some of the trade-offs faced when tuning~$D$.
\end{remark}

In the next section, we present a method to improve the regression process by carefully generating the dataset and propose a possible parameterization of $D$, before illustrating the merits of the criterion in Sec.~\ref{sec:results}.

%% file: sections/3_numerical_methods.tex
\section{Numerical methods}
\label{sec:numerical_methods}

As in~\citet{Ramos2020}, we approximate the transformation~$\Tinv$ by a neural network\footnote{Note that $\T$ can also be approximated using the same methodology, but is not necessary for state estimation.} of weights $\invparam$. The resulting observer is illustrated in Fig.~\ref{fig:t_inv}: we feed the measurement $y$ into the observer dynamics $\eqref{eq:dynz}$, then apply the neural network model $\Tinvparam$.
To train $\Tinvparam$, i.e. perform nonlinear regression, a dataset of $N$ pairs $(x_i,z_i)$, $i \in \{1, \dots, N\}$ needs to be generated from trajectories of~\eqref{eq:main_dynamics},~\eqref{eq:dynz}.
The construction of this dataset poses an important challenge, as the observer state~$z$ converges towards~$\T(x)$ only after a transient period whose length depends on $D$. This transient is not suitable for gathering data to learn the transformation, since we do not have $x \simeq \Tinv(z)$ during the transient.
However, for autonomous nonlinear systems, the trajectories tend to converge towards the~$\omega$--limit sets~\citep{rouche1977stability} of the dynamics, so that the points $(x_i,z_i)$ after the transient tend to be close to these~$\omega$--limit sets, leading to an uninformative dataset. We solve this problem in Sec.~\ref{subsec:backward_forward_sampling}.

Further, in order to calibrate the observer using the gain tuning criterion~\eqref{eq:gain_tuning_criterion}, a parametrization of the gain matrix $D$ by a scalar $\Dparam$ is needed.
Then, one can either learn a model $\Tinvparam$ for each value of $\Dparam$ independently, or learn the transformation as a function of $\Dparam$. This yields a harder regression problem, but avoids needing to learn a new transformation each time the pair~$(D, F)$ is changed. This is discussed in Sec.~\ref{subsec:learn_dependency_D}.

\subsection{Backward-forward sampling}
\label{subsec:backward_forward_sampling}

The choice of $(x_i, z_i)$ pairs is critical to numerically approximate $\Tinv$.
In~\citet{Ramos2020}, inspired by~\citet{Praly_uniform_practical_output_regulation}, the authors propose to first generate an arbitrary grid of initial conditions $(x(0), z(0))$ using standard statistical methods such as Latin Hypercube Sampling (LHS). 
Then, relying on the observer's stability, meaning that it forgets its arbitrary initial condition $z(0)$ after some time, the dynamics $x(t)$ and $z(t)$ are simulated forward in time for $\tcut$, where $\tcut$ is chosen large enough such that $z(\tcut)$ is \say{close} to its steady-state. 
Finally, the beginning of the solutions $(x(t),z(t))$ for $t < \tcut$ is removed from the dataset.

\begin{figure}[t]
	\begin{center}
		\includegraphics[width=\columnwidth]{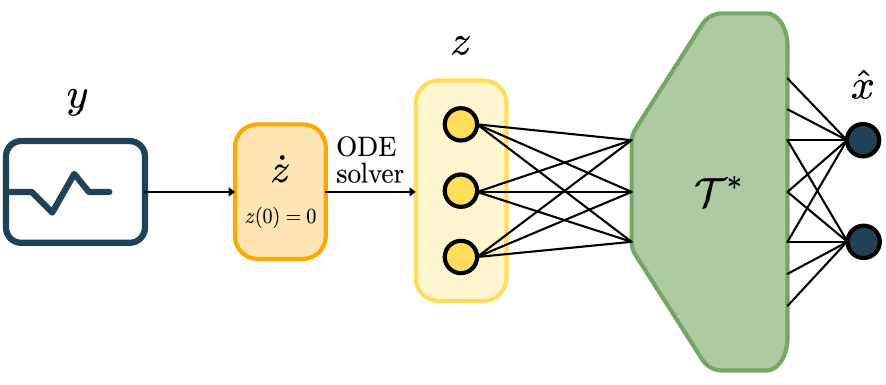} 
		\caption{Schematics of the learned KKL observer. First, we solve the ordinary differential equation~\eqref{eq:dynz} for the measurement $y$ generated from the original system~\eqref{eq:main_dynamics}. Then, the estimate $\hat{x}$ is computed as $\Tinvparam(z)$, where $\Tinvparam$ approximates $\Tinv$.}
		\label{fig:t_inv}
	\end{center}
\end{figure}

Unfortunately, this approach lets the dynamics dictate the position of the $(x_i, z_i)$ pairs: for large values of $\tcut$, they are bound to be located close to the~$\omega$--limit sets of the system~\citep{rouche1977stability}.
However, it is desirable to have training samples all over the state-space, especially in regions where the function $\Tinv$ are less smooth and therefore more difficult to approximate.

We propose the following methodology to generate an arbitrary dataset of $(x_i, z_i)$ pairs.

\begin{enumerate}
	\item Choose $N$ initial conditions $x_i(0) \in \X$, $i \in \{1, \dots, N\}$ using a uniform grid, LHS sampling, or any other method.
	
	\item Simulate the system $\dot{x} = f(x)$ from $x_i(0)$ \emph{backward in time} for $\tcut $ seconds.
	
	\item If the system diverges in backward finite time, then $f$ should be saturated smoothly outside of $\X$ as suggested in~\citet{Praly_existence_KKL_observer,Pauline_Luenberger_observers_nonautonomous_nonlin_systs}. An example of saturation is provided in Sec.~\ref{subsec:qube}.
	
	\item Simulate both systems $\dot{x} = f(x)$ and $\dot{z} = D z + F y$ with $y=h(x)$, starting from $x_i(-\tcut)$ obtained previously and $z_i(-\tcut) = z_0$, where $z_0$ is an arbitrary initial condition, for $\tcut$ seconds forward in time.
	
	\item Set the training dataset to $(x_i, z_i) = (x_i(0), z_i(0))$ as obtained from backward-forward simulation.
\end{enumerate}

With this approach, the user can set the training points $x_i$ a priori and obtain the corresponding $z_i$ without the system dynamics modifying the desired state-space grid.

\subsection{Parametrization of $D$}
\label{subsec:learn_dependency_D}

In order to evaluate the proposed gain tuning criterion, we parametrize $D$ by a scalar $\Dparam$.
Several parametrizations can be considered, for example choosing $D$ as a given diagonal matrix multiplied by a factor.
In this paper, we propose to use a $d_z$-order Bessel filter with cut-off frequency $2 \pi \omegac$, while $F = \mathbf{1}_{d_z \times d_y}$ is fixed to guarantee the controllability of $(D,F)$.
We choose $D$ by setting its eigenvalues to be the filter's poles.
For any set of poles $(p_1,\hdots, p_{n})$ where $p$ poles are real and $m$ poles are complex conjugates such that $n = p + 2m$, we choose $D$ as the following block-diagonal matrix: 
\begin{align}
    D &= \left(\begin{smallmatrix} D_1 & \cdots & 0\\\vdots&\ddots&\vdots\\ 0 & \cdots&D_{p+m} \end{smallmatrix}\right), 
    D_i = \begin{cases} p_i &\text{if }p_i \text{ is real}\\ 
                        \left(\begin{smallmatrix}
                                \Re{p_i}&\Im{p_i}\\-\Im{p_i}&\Re{p_i}
                        \end{smallmatrix}\right)&\text{ otherwise}
        \end{cases}
        \label{eq:D}
\end{align}
The choice of parametrization influences the performance of the obtained models; analyzing these different possibilities further could be an interesting topic for future work.

We can then compute the gain tuning criterion \eqref{eq:gain_tuning_criterion} for different values of $\Dparam$, by learning a model $\Tinvparam$ for each value of interest.
However, this requires training several neural networks independently for each value of $D$, which can be tedious.
Also, if the observer needs to be fine-tuned, a new model will be required.
Instead, it is also possible to treat $\omegac$ as an extra input to the network, so that the transformation to approximate is $\Tinvparam(z, w_c)$. 
This yields a harder regression problem, so that training will require more data and a careful design, but also a single model for all values of $D$.
The user can then choose an acceptable value of $D$ for the use case at hand and directly use the previous model.
Alternatively, they can train again for this specific value of $D$ to obtain a more accurate approximation for this particular choice.
This approach can be advantageous for low-dimensional problems or when the observer will be needed in different experimental conditions without re-training.

In the next section, we illustrate the relevance of criterion~\eqref{eq:gain_tuning_criterion} for choosing $D$ in numerical simulations, using the proposed sampling scheme and parametrization of $D$.

%% file: sections/4_results.tex
\section{Results}
\label{sec:results}

We now evaluate the proposed approach on simulations of two nonlinear systems\footnote{Note that for our empirical criterion~\eqref{eq:gain_tuning_criterion} to be meaningful, the variables should be normalized~\citep{Skogestad_multivariate_feedback_control}. In these academic examples, the variables can be considered scaled.}. 
We demonstrate that $D$ can be tuned a posteriori by optimizing a metric such as~\eqref{eq:gain_tuning_criterion}, and show that it is a relevant criterion for choosing $D$ so as to limit the noise sensitivity of the state estimate.
We learn the observer as a function of $\omegac$ for the first system, independently for different values of $\omegac$ for the second system.
Note that the model can eventually be trained again after selecting $\omegac$ to reach higher accuracy.\footnote{Code for reproducing the results is available at \url{https://github.com/Centre-automatique-et-systemes/learn_observe_KKL.git}.}

\begin{figure}[t]
	\begin{center}
		\includegraphics[width=0.85\columnwidth]{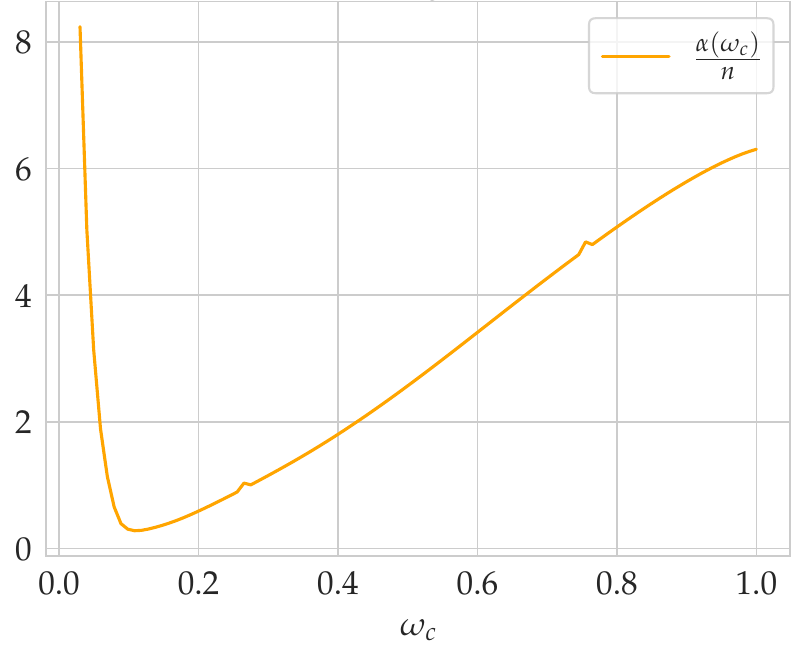} 
		\caption{Proposed gain tuning criterion \eqref{eq:gain_tuning_criterion} for the reverse Duffing oscillator, divided by $n=10,000$ points used to compute $\frac{\partial \Tinvparam}{\partial z}(z_j)$. The infinity and $H_2$ norms are high for low values of $\omegac$, while the gradient of the approximate transformation is high for high values. Choosing $\omegac = 0.15$ appears to be optimal with respect to this metric.}
		\label{fig:rev_duffing_sensitivity}
	\end{center}
\end{figure}

\subsection{Reverse Duffing oscillator}
\label{subsec:revDuffing}

\begin{figure*}[t]
\addtocounter{subfigure}{-3} 
\captionsetup[subfigure]{labelformat=empty}
	\begin{center}
		\subfloat{\includegraphics[width=0.33\textwidth]{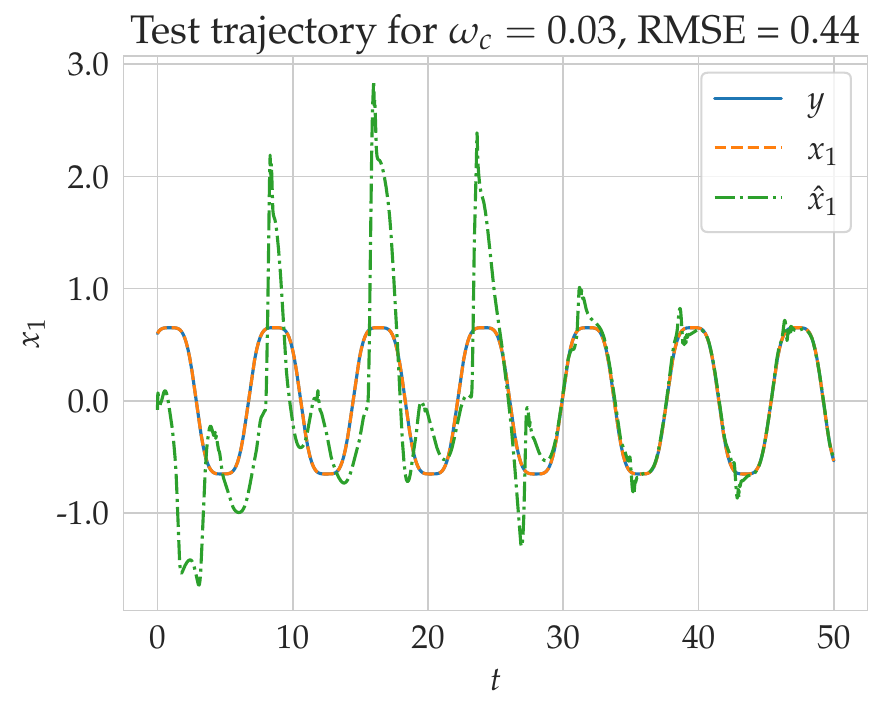}}
		\subfloat{\includegraphics[width=0.33\textwidth]{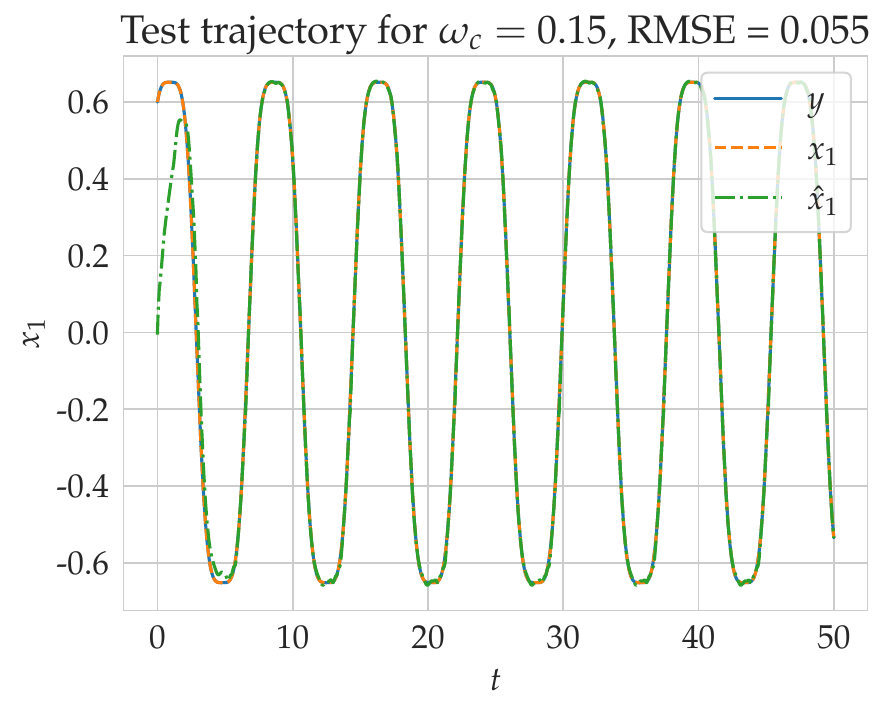}}
		\subfloat{\includegraphics[width=0.33\textwidth]{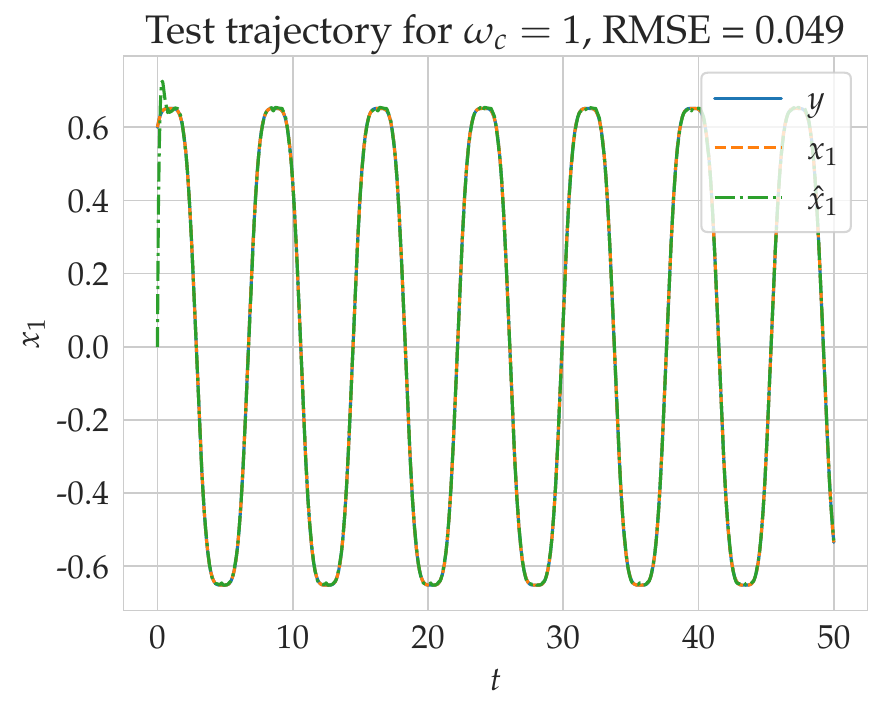}}
		\\
		\vspace{-1em}
		\subfloat[$\omegac=0.03$]{\includegraphics[width=0.33\textwidth]{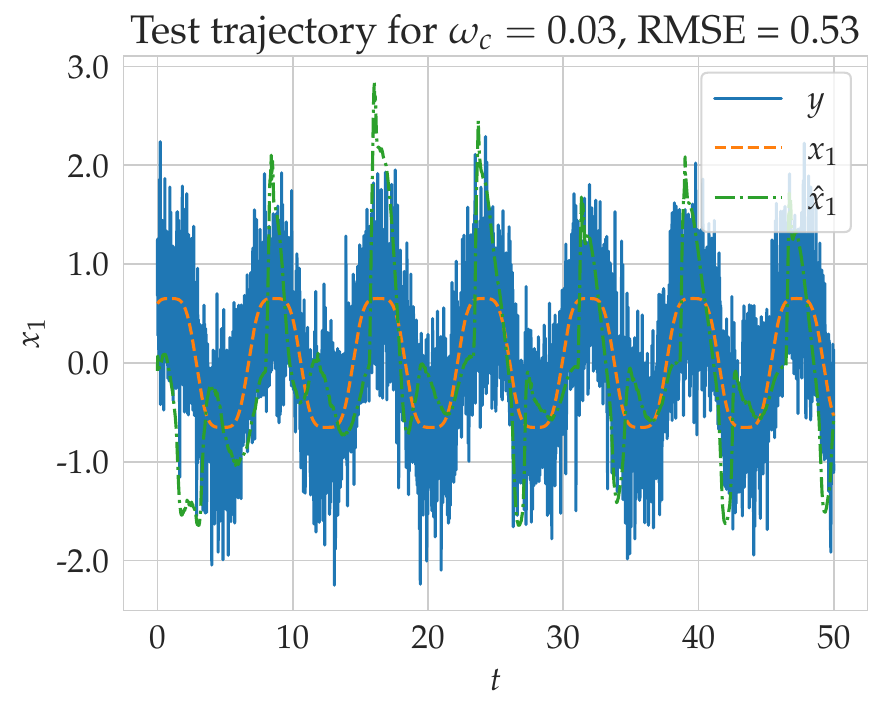}}
		\subfloat[$\omegac=0.15$]{\includegraphics[width=0.33\textwidth]{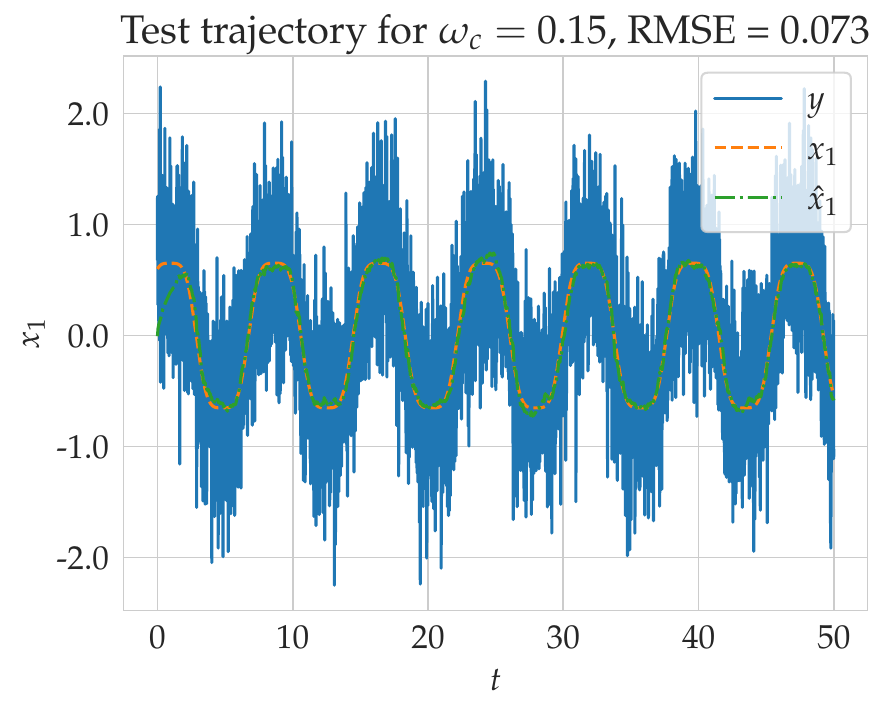}}
		\subfloat[$\omegac=1$]{\includegraphics[width=0.33\textwidth]{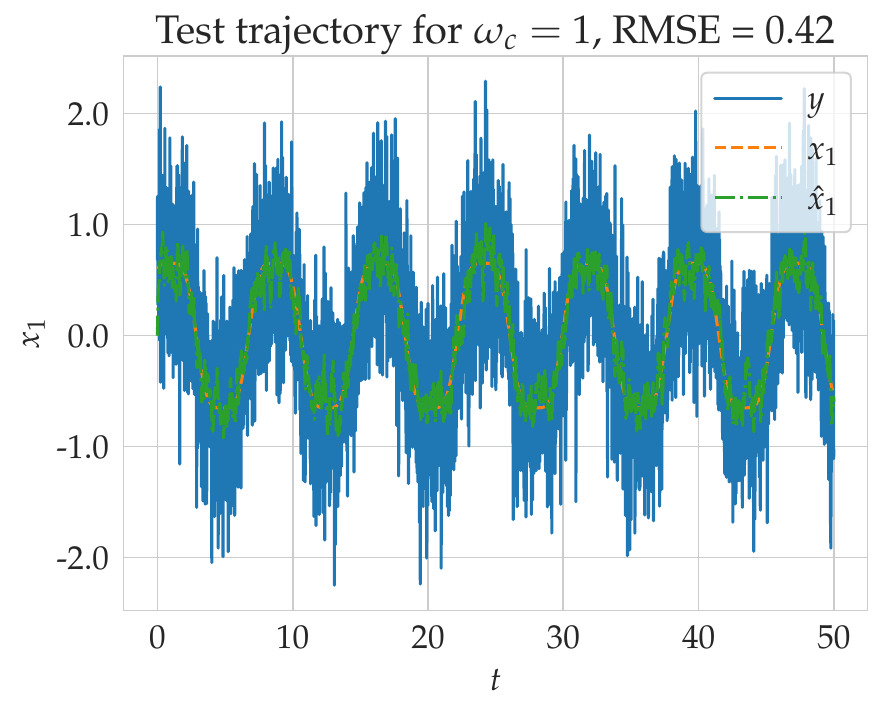}}
		\caption{Estimated trajectories of the reverse Duffing oscillator for $x(0)=(0.6,0.6)$, without measurement noise at the top, with noise $\mathcal{N}(0, 0.5)$ at the bottom. For each setting, we compute the Root Mean Squared Error (RMSE) over the whole trajectory. For low $\omegac$ (left), we observe long transients. For high $\omegac$ (right), the estimate is sensitive to high-frequency noise. For $\omegac = 0.15$ (middle), it is accurate and relatively robust to measurement noise.}
		\label{fig:rev_duffing_test_trajs}
	\end{center}
\end{figure*}

The reverse Duffing oscillator 
\begin{align}
	\label{eq:rev_duffing}
	\begin{cases}
		\dot{x}_1 = x^3_2 \\
		\dot{x}_2 = -x_1
	\end{cases} 
	\qquad
	y = x_1 
\end{align}
is a nonlinear system whose solutions evolve on invariant compact sets.
We choose a set of hundred values of~$\omega_{c_i}$ in $[0.03, 1]$.
Then, LHS is used to select $N=5,000$ samples $x_i \in [-1,1]^2$ for each value of $\omegac$.
The corresponding $z_i$ samples are computed using backward-forward sampling as described in Sec.~\ref{subsec:backward_forward_sampling}.
The training data is normalized to ease the optimization process.
For each given $\omegac$, $D$ is computed following~\eqref{eq:D}, while $F = \left(\begin{smallmatrix} 1 & 1 & 1\end{smallmatrix}\right) ^\top$.
The time $\tcut$ after which we consider that the observer has converged is set to $\frac{10}{\lambda_{min}(D)}$, where $\lambda_{min}(D)$ is the minimum absolute value of the real part of the eigenvalues of $D$, such that it is different for each value of $\omegac$.
The neural networks are multi-layer perceptrons with five hidden layers of size 50 and \emph{SiLU} activation, which is Lipschitz continuous and shows good performance.
We train $\Tinvparam$ by minimizing
\begin{align}
    \label{eq_train_loss}
    L(\invparam) & = \frac{1}{2} \sum_{x_i, z_i, \omega_{c_i}} \normltwo{ x_i - \Tinvparam(z_i, \omega_{c_i} ) }^2 .
\end{align}
We approximate $\Tinv(z, \omegac)$ over the training data as a function of $\omegac$, then compute the criterion~\eqref{eq:gain_tuning_criterion} for each value of $\omegac$ over a uniform grid of $n=10,000$ test points $z_j$, also obtained with backward-forward sampling.
The empirical criterion is shown in Fig.~\ref{fig:rev_duffing_sensitivity}.


The choice of $\omegac$ greatly influences the performance of the learned observer, as seen in Fig.~\ref{fig:rev_duffing_test_trajs}. 
In our simulations, lower values of $\omegac$ lead to a long convergence time and large overshoot, which corresponds to high values of $\normHtwo{G_z}$.
However, low $\omegac$ also yields a high signal to noise ratio in $z$, such that the observer is relatively robust to measurement noise. This is illustrated in the left column of Fig.~\ref{fig:rev_duffing_test_trajs}.
On the other hand, high values of $\omegac$ lead to a high gradient of $\Tinvparam$: the approximate transformation is not smooth and therefore very sensitive to changes in $z$, hence to measurement noise. The signal to noise ratio in $z$ is also low due to the fast eigenvalues of $D$. This is depicted at the bottom right of Fig.~\ref{fig:rev_duffing_test_trajs}. 
In the central column of Fig.~\ref{fig:rev_duffing_test_trajs}, we select $\omegac = 0.15$ the optimal value according to criterion~\eqref{eq:gain_tuning_criterion}. This setting yields an acceptable trade-off between these different aspects: both overshoot and noise sensitivity remain limited.
Hence, the proposed gain turning criterion leads to satisfying performance for this use case.

\subsection{Quanser Qube}
\label{subsec:qube}

We then consider simulations of a rotational inverted pendulum: the Qube Servo 2 by~\citet{QuanserQube}, illustrated in Fig.~\ref{fig:qube_picture}.
Its state of dimension four consists of two angles $(\theta_1, \theta_2)$ and two angular velocities $(\dot\theta_1, \dot\theta_2)$; we measure $y=\theta_1$.
Its trajectories diverge in finite backward time. Hence, as suggested in \citet{Praly_existence_KKL_observer, Pauline_Luenberger_observers_nonautonomous_nonlin_systs}, we consider the modified system 
\begin{align}
    \label{eq:van_der_pol}
    \dot{x} & = f(x) g(x) ,
    \notag
    \\
    g(x) & = \begin{cases}
        1 & \text{if } \normltwo{x} \leq r
        \\
        0  & \text{if } \normltwo{x} \geq r + d
        \\
        p(\normltwo{x}-r) & \text{otherwise}
    \end{cases}
\end{align}
where $f$ is the dynamics model of the Qube and $g$ is a saturation function. We set $r=50$ and $d=100$.
The function $p(\cdot)$ is a polynomial of order three chosen such that $g$ be $C^1$.
This modified system has the same trajectories as the original system inside $\X$ but does not blow up in backward time from any initial condition in $\X$.

Due to the curse of dimensionality, a large amount of data is necessary to learn $\Tinv$ with $d_x=4$. In order to limit the computations, we generate data along realistic trajectories for the autonomous pendulum: we select $500$ samples in a hypercube around the upward equilibrium position, use backward-forward sampling to obtain the corresponding $z$ values, then run a joint simulation of both the $x$ and $z$ trajectories for $8$s, sampled with time steps of $0.04$s. This leads to $N=10^6$ points for each of $41$ values of $\omegac \in [1,5]$, for which we learn one model $\Tinvparam$ independently.

We then compute the empirical criterion~\eqref{eq:gain_tuning_criterion} for each value of $\omegac$ independently on a grid of $n=50,000$ points and obtain Fig.~\ref{fig:qube_sensitivity}. The minimum is reached at $\omegac=1.9$, which again seems to be a good compromise between long transients and sensitivity to measurement noise. This is illustrated in Fig.~\ref{fig:qube_test_traj_error}: high values of $\omegac$ lead to sensitivity to high frequency measurement noise, low values to long transients whenever the estimate is off, and $\omegac=1.9$ to an acceptable trade-off.

These results constitute a first step towards gain tuning for nonlinear observers. They can be considered as a proof of concept, showing that it is possible to tune the gains of KKL observers by parametrizing the gain matrix with a scalar $\omegac$ then optimizing this scalar w.r.t. certain metrics. Note that many such metrics could be considered depending on the use case at hand. We propose the gain tuning criterion~\eqref{eq:gain_tuning_criterion}, which displays relevant aspects of the trade-off faced when choosing $D$ as illustrated by our results, but other quantities could also be helpful.

\begin{figure}[t]
	\begin{center}
		\includegraphics[width=0.3\columnwidth]{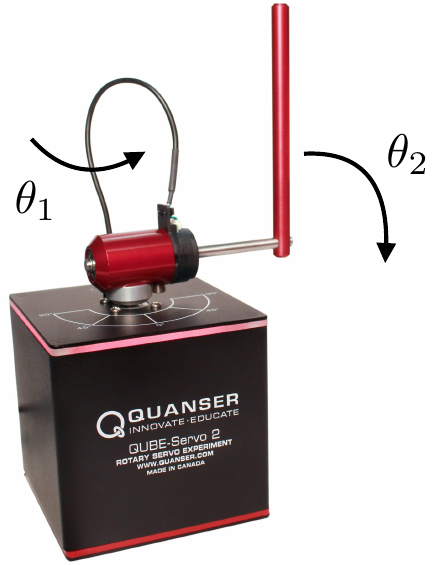} 
		\caption{Qube Servo 2 by~\citet{QuanserQube}.}
		\label{fig:qube_picture}
	\end{center}
\end{figure}

\begin{figure}[t]
	\begin{center}
		\includegraphics[width=0.9\columnwidth]{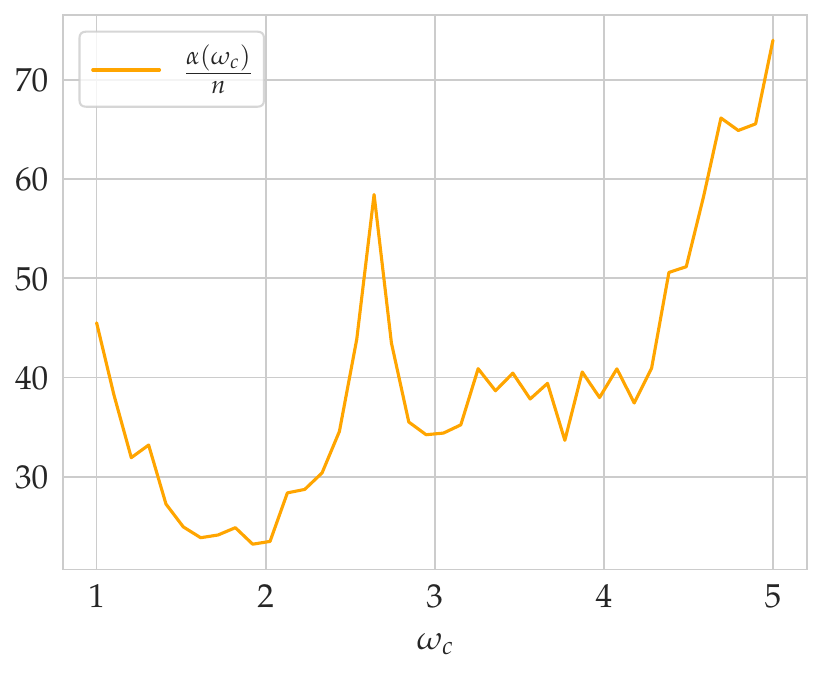} 
		\caption{Proposed gain tuning criterion~\eqref{eq:gain_tuning_criterion} for the Qube, divided by $n=10,000$. Choosing $\omegac = 1.9$ appears to be optimal with respect to this metric.}
		\label{fig:qube_sensitivity}
	\end{center}
\end{figure}

\begin{figure}[t]
	\begin{center}
		\includegraphics[width=\columnwidth]{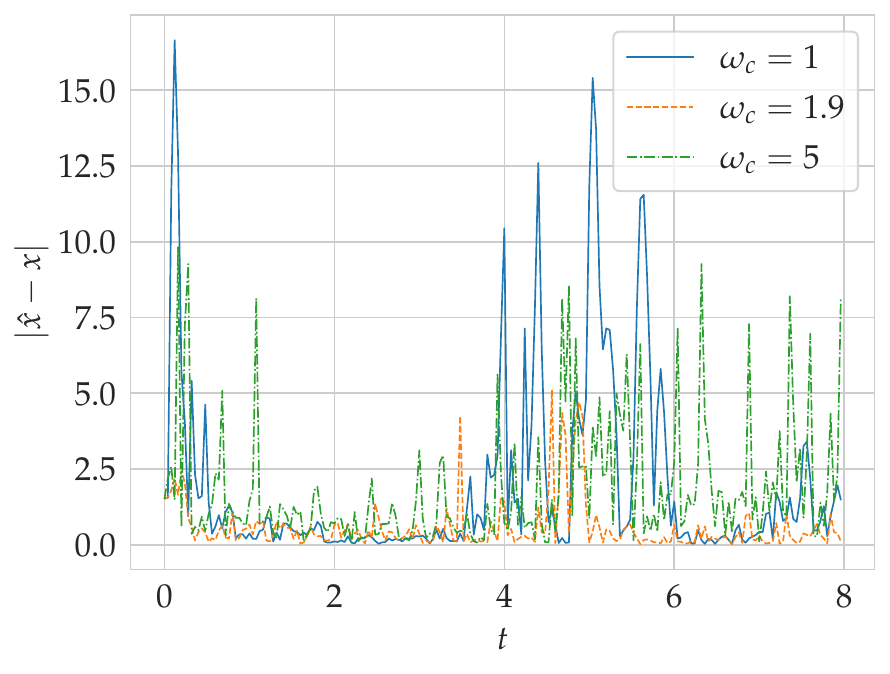} 
		\caption{Estimation RMSE $\normltwo{\hat{x}-x}$ for a simulated test trajectory of the Qube starting at $x(0)=(0.1,0.1, 0,0)$, with Gaussian noise $\mathcal{N}(0, 0.025)$ on the measurement.
		We observe a high sensitivity to noise for high values of $\omegac$ (RMSE $= 2.8$ for $\omegac=5$). In contrast, the sensitivity to noise is lower for low values of $\omegac$, but we observe long transients, which lead to RMSE $= 3.7$ for $\omegac=1$. The value $\omegac=1.9$ compromises between performance and sensitivity to noise (RMSE $= 1$).}
		\label{fig:qube_test_traj_error}
	\end{center}
\end{figure}

\subsection{Discussion on the numerical results}

Independently of the chosen parametrization of $D$ or the strategy for generating the training data, approximating $\Tinv$ with neural networks bears the risk of overfitting. This risk is limited by monitoring the training loss~\eqref{eq_train_loss} on a validation set and stopping the training early if this validation loss starts increasing, along with other standard techniques to restrict overfitting in supervised learning. For the Quanser Qube in Sec.~\ref{subsec:qube}, due to the difference between the simulation model used for generating the training data and the hardware, the performance of the learned observer decreases when used on experimental rather than simulated test trajectories. This illustrates that overfitting $\Tinvparam$ to the particular system at hand, i.e. the robustness of the numerical KKL observer to model error, remains an issue. Recent works propose performance improvements of such observers~\citep{Johansson_learning_based_KKL, Miao_learning_robust_state_observers_NODEs}, which could help alleviate this. Fine-tuning the learned observer based on measurements could also help adapt it to the physical system at hand, for instance by retraining it inside an output predictor as in~\citet{Andrieu_deepKKL_output_prediction_nonlin_systs}.


%% file: sections/6_conclusion.tex
\section{Conclusion and perspectives}
\label{sec:conclu}

In this paper, we tackle the problem of gain tuning for KKL observers of autonomous nonlinear systems. We propose to numerically approximate the observer from simulation data, as introduced in~\citet{Ramos2020}, with an improved backward-forward sampling scheme. 
We parametrize the observer dynamics matrix $D$ with a scalar $\omega_c$, derive an empirical criterion for tuning it, and demonstrate on two numerical examples that it encompasses some relevant aspects of its influence on the performance. 
We propose either to learn an observer for each value of $\omega_c$ of interest, or to directly learn a family of models that also takes this parameter as an input.


Similarly to~\citet{Nadri_autoencoder_KKL_discrete_time,Lusch_deeplearning_universal_linear_embeddings_nonlin_dyns}, it is also possible to learn a model of $\T$ and $\Tinv$ jointly using an autoencoder structure, such that the latent variable $z$ verifies~\eqref{eq:dynz}. The cost function is then made up of a reconstruction loss and a loss on the PDE~\eqref{eq:pdeT} verified by $\T$, such that an invertible solution to~\eqref{eq:pdeT} is approximated on a grid of samples of $x$. 
This approach enables the user to optimize $D$ jointly with the models $\Tparam$, $\Tinvparam$ and to add terms to the cost functions to penalize other aspects, such as the criterion~\eqref{eq:gain_tuning_criterion}. However, it is also harder to train than the supervised approach.
Further research aims at improving the accuracy of learning-based KKL observers, such as~\cite{Johansson_learning_based_KKL}.

Many other questions remain open. As often in machine learning, it is unclear how to sample the state-space to generate the dataset optimally. Iterative active learning procedures can be envisioned, for example by learning the observer, then resampling in the parts of the state-space with the highest error, and learning again until the desired accuracy is achieved everywhere. Selecting the state-space grid a priori to achieve a given accuracy on the transformations could also be considered, as investigated in~\citet{Praly_uniform_practical_output_regulation}.
Extending KKL observers to nonautonomous systems is investigated in~\citet{Pauline_Luenberger_observers_nonautonomous_nonlin_systs}; adapting the learning-based methodology to such systems is also a topic for future research.

%% file: root.bbl
\begin{thebibliography}{26}
\providecommand{\natexlab}[1]{#1}
\providecommand{\url}[1]{\texttt{#1}}
\providecommand{\urlprefix}{URL }
\expandafter\ifx\csname urlstyle\endcsname\relax
  \providecommand{\doi}[1]{doi:\discretionary{}{}{}#1}\else
  \providecommand{\doi}{doi:\discretionary{}{}{}\begingroup
  \urlstyle{rm}\Url}\fi

\bibitem[{Andrieu and Praly(2006)}]{Praly_existence_KKL_observer}
Andrieu, V. and Praly, L. (2006).
\newblock {On the existence of a Kazantzis-Kravaris/Luenberger observer}.
\newblock \emph{SIAM Journal on Control and Optimization}, 45(2), 422--456.

\bibitem[{Astolfi et~al.(2018)Astolfi, Marconi, Praly, and Teel}]{Astolfi2018}
Astolfi, D., Marconi, L., Praly, L., and Teel, A.R. (2018).
\newblock Low-power peaking-free high-gain observers.
\newblock \emph{Automatica}, 98, 169--179.

\bibitem[{Bernard(2019)}]{Pauline_observer_design_nonlin_systs}
Bernard, P. (2019).
\newblock \emph{{Observer Design for Nonlinear Systems}}.
\newblock Springer International Publishing.

\bibitem[{Bernard and
  Andrieu(2019)}]{Pauline_Luenberger_observers_nonautonomous_nonlin_systs}
Bernard, P. and Andrieu, V. (2019).
\newblock {Luenberger Observers for Nonautonomous Nonlinear Systems}.
\newblock \emph{IEEE Transactions on Automatic Control}, 64(1), 270--281.

\bibitem[{Bernard et~al.(2022)Bernard, Andrieu, and
  Astolfi}]{Pauline_survey_observers}
Bernard, P., Andrieu, V., and Astolfi, D. (2022).
\newblock {Observer Design for Continuous-Time Dynamical Systems}.
\newblock \emph{Annual Reviews in Control}, 53, 224--248.

\bibitem[{Bornard and Hammouri(1991)}]{Bornard1991}
Bornard, G. and Hammouri, H. (1991).
\newblock A high gain observer for a class of uniformly observable systems.
\newblock In \emph{Proceedings of the 30th IEEE Conference on Decision and
  Control}, 1494 --1496.

\bibitem[{Gelb(1974)}]{Gelb1974}
Gelb, A. (1974).
\newblock \emph{Applied optimal estimation}.
\newblock MIT press.

\bibitem[{Henwood(2014)}]{Henwood_PhD_estimation_online_parametres_machines_electriques_suivi_temperature_composants}
Henwood, N. (2014).
\newblock \emph{{Estimation en ligne de param{\`{e}}tres de machines
  electriques pour v{\'{e}}hicule en vue d'un suivi de la temp{\'{e}}rature de
  ses composants}}.
\newblock Ph.D. thesis, Mines ParisTech.

\bibitem[{Janny et~al.(2021)Janny, Andrieu, Nadri, and
  Wolf}]{Andrieu_deepKKL_output_prediction_nonlin_systs}
Janny, S., Andrieu, V., Nadri, M., and Wolf, C. (2021).
\newblock {Deep KKL: Data-driven Output Prediction for Non-Linear Systems}.
\newblock In \emph{Proceedings of the IEEE Conference on Decision and Control}.

\bibitem[{Kalman and Bucy(1961)}]{Kalman1961}
Kalman, R.E. and Bucy, R.S. (1961).
\newblock New results in linear filtering and prediction theory.
\newblock \emph{Journal of Basic Engineering}, 83, 95--108.

\bibitem[{Kazantzis and Kravaris(1998)}]{Kazantzis1998}
Kazantzis, N. and Kravaris, C. (1998).
\newblock Nonlinear observer design using lyapunov’s auxiliary theorem.
\newblock \emph{Systems \& Control Letters}, 34(5), 241--247.

\bibitem[{Khalil and Praly(2014)}]{Khalil2014}
Khalil, H.K. and Praly, L. (2014).
\newblock High-gain observers in nonlinear feedback control.
\newblock \emph{International Journal of Robust and Nonlinear Control}, 24(6),
  993--1015.

\bibitem[{Krener(2003)}]{Krener2003}
Krener, A.J. (2003).
\newblock {The convergence of the extended Kalman filter}.
\newblock In \emph{Directions in mathematical systems theory and optimization},
  173--182. Springer.

\bibitem[{Luenberger(1966)}]{Luenberger1966}
Luenberger, D. (1966).
\newblock Observers for multivariable systems.
\newblock \emph{IEEE Transactions on Automatic Control}, 11(2), 190--197.

\bibitem[{Lusch et~al.(2018)Lusch, Kutz, and
  Brunton}]{Lusch_deeplearning_universal_linear_embeddings_nonlin_dyns}
Lusch, B., Kutz, J.N., and Brunton, S.L. (2018).
\newblock {Deep learning for universal linear embeddings of nonlinear
  dynamics}.
\newblock \emph{Nature Communications}, 9(1).

\bibitem[{Maggiore and Passino(2003)}]{maggiore2003separation}
Maggiore, M. and Passino, K.M. (2003).
\newblock {A separation principle for a class of non-UCO systems}.
\newblock \emph{IEEE Transactions on Automatic Control}, 48(7), 1122--1133.

\bibitem[{Marconi and Praly(2008)}]{Praly_uniform_practical_output_regulation}
Marconi, L. and Praly, L. (2008).
\newblock {Uniform practical nonlinear output regulation}.
\newblock \emph{IEEE Transactions on Automatic Control}, 53(5), 1184--1202.

\bibitem[{Miao and Gatsis(2022)}]{Miao_learning_robust_state_observers_NODEs}
Miao, K. and Gatsis, K. (2022).
\newblock {Learning Robust State Observers using Neural ODEs (longer version)}.
\newblock \emph{Preprint arXiv:2212.00866}.

\bibitem[{Niazi et~al.(2022)Niazi, Cao, Sun, Das, and
  Johansson}]{Johansson_learning_based_KKL}
Niazi, M.U.B., Cao, J., Sun, X., Das, A., and Johansson, K.H. (2022).
\newblock {Learning-based Design of Luenberger Observers for Autonomous
  Nonlinear Systems}.
\newblock \emph{Preprint arXiv:2210.01476}.

\bibitem[{Peralez and Nadri(2021)}]{Nadri_autoencoder_KKL_discrete_time}
Peralez, J. and Nadri, M. (2021).
\newblock {Deep Learning-based Luenberger observer design for discrete-time
  nonlinear systems}.
\newblock In \emph{Proceedings of the IEEE Conference on Decision and Control},
  4370--4375. IEEE.

\bibitem[{Quanser(2022)}]{QuanserQube}
Quanser (2022).
\newblock Quanser courseware and resources.
\newblock \urlprefix\url{https://www.quanser.com/products/qube-servo-2/}.

\bibitem[{Ramos et~al.(2020)Ramos, Meglio, Morgenthaler, da~Silva, and
  Bernard}]{Ramos2020}
Ramos, L.D.C., Meglio, F.D., Morgenthaler, V., da~Silva, L.F.F., and Bernard,
  P. (2020).
\newblock {Numerical design of Luenberger observers for nonlinear systems}.
\newblock In \emph{Proceedings of the 59th IEEE Conference on Decision and
  Control}, 5435--5442.

\bibitem[{Rouche et~al.(1977)Rouche, Habets, and Laloy}]{rouche1977stability}
Rouche, N., Habets, P., and Laloy, M. (1977).
\newblock \emph{Stability theory by Liapunov's direct method}.
\newblock Springer.

\bibitem[{Scaman and Virmaux(2018)}]{Scaman_Lipschitz_deepNN_AutoLip}
Scaman, K. and Virmaux, A. (2018).
\newblock {Lipschitz regularity of deep neural networks: Analysis and efficient
  estimation}.
\newblock \emph{Advances in Neural Information Processing Systems 32},
  3835--3844.

\bibitem[{Skogestad and
  Postlethwaite(2005)}]{Skogestad_multivariate_feedback_control}
Skogestad, S. and Postlethwaite, I. (2005).
\newblock \emph{{Multivariable Feedback Control: Analysis and Design}}.
\newblock John Wiley {\&} Sons.

\bibitem[{Toivonen(2010)}]{Toivonen_signal_system_norms}
Toivonen, H. (2010).
\newblock {Signal and system norms}.
\newblock \emph{Lecture Notes for the Course "Advanced Control Methods"}.
\newblock \urlprefix\url{http://users.abo.fi/htoivone/courses/robust/rob2.pdf}.

\end{thebibliography}
